\documentclass{article}

\usepackage{microtype}
\usepackage{graphicx}
\usepackage{subcaption}
\usepackage{booktabs} 
\usepackage{xspace}
\usepackage{enumitem}
\usepackage{placeins}

\usepackage[noend]{algpseudocode}

\usepackage{hyperref}

\usepackage[preprint]{icml2026}

\usepackage{amsmath}
\usepackage{amssymb}
\usepackage{mathtools}
\usepackage{amsthm}

\usepackage[capitalize,noabbrev]{cleveref}

\theoremstyle{plain}
\newtheorem{theorem}{Theorem}[section]
\newtheorem{proposition}[theorem]{Proposition}

\theoremstyle{definition}

\theoremstyle{remark}

\usepackage[textsize=tiny]{todonotes}

\icmltitlerunning{Transformer for discrete fair division}

\definecolor{attn}{RGB}{77,136,255}
\definecolor{normc}{RGB}{255,179,71}
\definecolor{gluc}{RGB}{102,204,153}
\definecolor{proj}{RGB}{186,145,255}
\definecolor{outc}{RGB}{255,102,102}
\definecolor{aux}{RGB}{160,160,160}

\newcommand{\FF}{\textsc{FairFormer}\xspace}
\newcommand{\Items}{\mathit{\Omega}}

\newcommand{\SAttn}{\textbf{FFSelfAttn}\xspace}
\newcommand{\CAttn}{\textbf{FFCrossAttn}\xspace}

\newcommand{\F}{\texttt{false}}
\newcommand{\T}{\texttt{true}}
\newcommand{\UW}{\texttt{UW}\xspace}
\newcommand{\NW}{\texttt{NW}\xspace}

\newcommand{\EFoneR}{\textsc{EF1-Quick-Repair}\xspace}

\newcommand{\unif}{\mathbb{U}}
\newcommand{\norm}{\mathrm{\textsc{Norm}}}

\begin{document}

\twocolumn[
  \icmltitle{\FF: A transformer architecture for discrete fair division}



  \icmlsetsymbol{equal}{*}

  \begin{icmlauthorlist}
    \icmlauthor{Chris Mascioli}{yyy}
    \icmlauthor{Satyam Goyal}{yyy}
    \icmlauthor{Mithun Chakraborty}{yyy}
  \end{icmlauthorlist}

  \icmlaffiliation{yyy}{Department of Computer Science, University of Michigan, United States of America}

  \icmlcorrespondingauthor{Chris Mascioli}{cmasciol@umich.edu}

  \icmlkeywords{Machine Learning, ICML}

  \vskip 0.3in
]



\printAffiliationsAndNotice{}  

\begin{abstract}
We propose a deep neural network-based solution to the problem of 
allocating indivisible goods under additive subjective valuations without monetary transfers, trading off economic efficiency with envy-based fairness. 
We introduce \FF, an amortized, permutation-equivariant two-tower transformer that encodes items and agents as unordered token sets, applies self-attention within each set, and uses item-to-agent cross-attention to produce per-item assignment distributions in a single forward pass.
\FF is trained end-to-end to maximize expected log-Nash welfare on sampled instances, requiring no solver supervision, unrolled allocation procedures, or fairness labels.
At test time, we discretize by row-wise $\arg\max$ and apply a lightweight post-processing routine that transfers items to eliminate violations of envy-freeness up to one item while prioritizing improvements in  Nash welfare.
Our approach generalizes beyond its training regime and achieves near-optimal welfare (e.g., for uniformly sampled valuations, $96$--$97\%$ for Nash welfare; $95$--$96\%$ for utilitarian welfare), outperforming strong baselines in solution quality and/or runtime.
\end{abstract}

\section{Introduction}
The fair distribution of subjectively valued scarce resources is an age-old problem of widespread importance \cite{brams1996fair}. Any rigorous solution approach must address questions on multiple levels: how to operationalize the normative concept of fairness; whether various fairness criteria are feasible individually, jointly with each other, and/or in conjunction with economic efficiency; and, if yes, whether we can efficiently compute allocations satisfying them. These questions become harder when the task is to allocate a collection of \textit{indivisible} items \cite{BouveretChMa16} such as non-shareable facilities, resource usage time-slots, or artworks, especially when monetary transfers and substitutions are prohibited or undesirable (e.g., inheritance, divorce settlements, distribution of donated food items among local food pantries, etc.). 
This task, often called \textit{discrete fair division} \cite{amanatidis2023fair}, has been the subject of 
active cross-disciplinary research  for a couple of decades  \cite{LiptonMaMo04,Budish11,gourves12014near,caragiannis2019unreasonable,amanatidis2021maximum}. 

Theoretical work on (discrete) fair division has produced a rich array of axioms and objectives that capture various perspectives on fairness. 
An important family of approaches is based on maximizing a \textit{(social) welfare function} \cite{sen2018collective}, a mapping of the vector of all beneficiaries' realized utilities under an allocation to a score that allows us to rank allocations. A prominent member of this family is the \textit{maximum Nash welfare} (MNW) principle that picks an allocation with the highest product (equivalently, geometric mean) of utilities. 
This MNW problem is NP-hard and, in fact, APX-hard \cite{akrami2022maximizing}; it has spawned extensive research on approximation and tractable special cases \cite{cole2015approximating,cole2017convex,barman2018finding,barman2018greedy}.

Another set of approaches involves formulating fairness constraints from scratch and devising allocation algorithms to satisfy them or their relaxations. This has resulted in concepts such as envy, proportionality, maximin share, etc.; see the survey by \citet{amanatidis2023fair} for details. We will focus only on the envy-based fairness paradigm.  An allocation is \textit{envy-free} (EF) if no agent values its own share in the allocation less than any other agent's share. Envy-freeness cannot be guaranteed for indivisible items (think of one valuable item and two claimants). EF1 \citep{Budish11,LiptonMaMo04} and EFX \citep{caragiannis2019unreasonable} are important relaxations of EF which stipulate that envy can be eliminated by hypothetically removing at least one item and any item respectively from an envied agent's allocated bundle of items. The existence of EFX in general is an open question but, as long as agents' valuations are monotone, an EF1 allocation always exists   and can be efficiently computed by the classic \textit{envy cycle elimination} algorithm \citep{LiptonMaMo04}. If valuations are additive, then an EF1 allocation can be obtained by a simple and practical \textit{round-robin picking sequence protocol} \cite{caragiannis2019unreasonable}, which allows each agent to pick its most preferred remaining item in each turn. However, these EF1-ensuring approaches do not provide any guarantees on welfare or efficiency. An interesting connection between the welfare-based and envy-based paradigms was established by \citet{caragiannis2019unreasonable} who proved that an MNW allocation under additive valuations is not only Pareto optimal but also EF1. An approximate solution to the MNW problem is not known to provide any guarantees on envy relaxation even in approximation \cite{amanatidis2023fair} but can still serve as a reasonable guiding principle for discrete fair division.


There is a nascent literature on \textit{learning} allocation rules that produce good solutions in one shot. \citet{mishra2022eef1} propose EEF1‑NN that utilizes a U‑Net–style \cite{ronneberger2015unet} architecture and a Lagrangian training objective to target utilitarian welfare subject to EF1, reporting high EF1 rates and near‑optimal welfare on synthetic data. Closer in spirit to algorithmic mechanisms, \citet{maruo2024learning} parameterize the classic round‑robin protocol with a differentiable relaxation, learning the agent order while guaranteeing EF1 by construction. We significantly push the boundaries of this research area by devising and systematically assessing a novel indivisible goods allocator under additive valuations that builds on the powerful transformer framework \cite{vaswani2017attention}.



\subsection{Our contributions}
We propose \FF(Fig.~\ref{fig:fairformer}; Section~\ref{sec:fairformer}), a symmetry-aware transformer that takes as input a discrete fair division instance represented by a valuation matrix and generates a complete allocation matrix in a single forward pass. It utilizes a set-structured design that treats agents and items as unordered sets,
applies \textit{self-attention} within each set, and fuses them via \emph{item-to-agent cross-attention} before a temperature-controlled
softmax projection produces per-item assignment probabilities. The temperature schedule smoothly interpolates from fractional allocations at the start of training to near one-hot allocations at the end,
providing a differentiable path to discrete decisions without surrogate labels. The salient properties of \FF are:
\begin{itemize}[leftmargin=*]
\item \textbf{Permutation symmetry.} \FF follows the principle of \textit{exchangeability}: permuting input agent or item indices should only shuffle the corresponding dimensions of the output allocation accordingly. This is achieved by a two-tower architecture that operates on agent and item token sets without positional encodings (unlike sequential transformers), maintaining \textit{permutation-equivariance} over both agents and items. It uses only exchangeable layers \cite{hartford2018deep} and residual stacks \cite{he2016deep} with multi-head attention and gated feed-forward layers. Here, attention serves not as a temporal mechanism but as a flexible, symmetry-preserving operator capable of modeling higher-order item-agent interactions. 
This use of attention as a \textit{structural primitive} aligns the model’s inductive bias with the underlying exchangeability of fair division instances.
\item \textbf{Direct welfare optimization.} We train \FF end-to-end to maximize the expected logarithm of Nash welfare with no solver targets or oracle supervision,
no procedural unrolling, and no fairness labels. Unlike previous work in this vein, \FF uses only the raw, natural matrix representation of input instances with no pre-processing, and does not specialize to a particular procedural rule (e.g., round-robin). 
\item \textbf{Amortized allocation.} After training, \FF produces allocations for unseen valuation matrices in a single forward pass, enabling constant-time inference with respect to the number of optimization iterations, amortizing \cite{amos2023tutorial} across instance sizes.
\end{itemize}




We complete our allocation pipeline by refining the discretized output of \FF with a post-processing routine that looks for violations of the EF1 condition so that it can transfer from envied to envious agent an item yielding the best Nash welfare improvement (Section~\ref{sec:ef1repair}); this \textit{EF1 repair} heuristic, rooted in theoretical work by \citet{caragiannis2019unreasonable}, guarantees an ex-post EF1 output with (empirically) low overhead.
Finally, we experimentally evaluate our approach and compare its performance with those of select baselines from the literature on fair division data sets with diverse statistical properties and varying numbers of items and agents (Section~\ref{sec:experiments}). Our metrics gauge welfare, envy-based fairness, and runtime. \FF with EF1 repair demonstrates impressive generalization beyond its training regime and near-optimal welfare values (even for utilitarian welfare for which it was not trained), outperforming baselines on several metrics.

All relevant code is available at: \url{https://anonymous.4open.science/r/fair-allocation-transformer-34A5/README.md}



\subsection{Further related work}



Our design follows a broader line of work on learning over sets and exchangeable structures. 
Deep Sets \cite{zaheer2017deep} and Set Transformers \cite{lee2019set} model interactions within sets using permutation-invariant or -equivariant operators. The NFGTransformer \cite{liu2024nfgtransformer} extends this principle to game-theoretic settings, leveraging invariance to the ordering of actions and strategies in normal-form games. 
These architectures demonstrate that respecting symmetry specific to the data domain 
yields improved generalization and interpretability. 
For relational structures such as valuation matrices, \citet{hartford2018deep} introduced a parameter-sharing scheme enforcing row- and column-wise permutation-equivariance, which directly motivates the two-tower design of \FF.
Deep learning has also been widely applied to auction design (with payments), targeting revenue or welfare under incentive constraints \cite{tacchetti2019_auctionNN, dutting2024optimal}.  
A recurring theme in this area is to encode permutation symmetries over bidders/items into the architecture. \citet{rahme2021permutationequivariant} build a permutation-equivariant network for auctions; \citet{duan2022context} and \citet{ivanov2022optimal} propose a transformer-based network to solve these problems. We also exploit exchangeability, but in a no‑payment fair‑division setting and for a different objective (NW vs.\ revenue), with an item-to-agent cross‑attention head that maps a valuation matrix to per‑item assignment distributions in a single forward pass.


\section{Technical preliminaries}\label{sec:prelim}

An \textit{instance} of discrete fair division is a tuple $\langle N, \Items, \{v_i\}_{i \in N} \rangle$ where $N = \{1,2,\dots,n\}$ is a set of $n$ \textit{agents}, $\Items=\{o_1,o_2,\dots,o_m\}$ is a set of $m$ \textit{items}, and $v_i:2^{\Items} \to [0,\infty)$ is a function expressing the \textit{valuation} of agent $i\in N$ for each \textit{bundle} (i.e., subset of items), for positive integers $n$ and $m$.  We will focus on settings where $m\ge n$ and all valuations are additive, i.e., $\forall i \in N. \forall S \subseteq \Items.v_i(S)=\sum_{o\in S} v_i(o)$, where $v_i(o)$ is standard shorthand for $v_i(\{o\})$.
For this class, an instance can be equivalently defined by an $m\times n$ valuation matrix $V$ with entries $V_{o,i}=v_i(o)$, $i\in N$, $o \in \Items$. 

Given an instance, an \textit{allocation} $A=\{A_1,A_2,\dots,A_n\}$ is an $n$-partition of $\Items$ (i.e., $\cup_{i\in N} A_i = \Items$; $A_i \cap A_j = \emptyset$ for every $(i,j)$ with $i \neq j$),
where $A_i$ denotes the bundle received by agent $i$. We call $v_i(A_i)$ the \textit{realized} valuation of agent $i$ under allocation $A$. We will represent an allocation $A$ as an $m\times n$ binary matrix with one-hot row vectors: the entries are $A_{o,i}\in\{0,1\}$, $i\in N$, $o\in \Items$, such that, for every $o$, $A_{o,i} = 1$ if and only if item $o$ is assigned to agent $i$. As part of our transformer-based solution approach, we will admit fractional allocation matrices whose entries satisfy the following criteria: $A_{o,i}\in [0,1]$, $i\in N$, $o\in \Items$, and $\sum_{i=1}^n A_{o,i} = 1$ for every $o$. An agent $i$'s bundle $A_i$ under any (potentially fractional) allocation $A$ is identified with the $i^\mathrm{th}$ column of the matrix $A$, and agent $i$'s (additive) valuation of any agent $j$'s bundle is $v_i(A_j)=\sum_{k=1}^m V_{o_k,i} A_{o_k,j}$.

\subsection{Efficiency and fairness criteria}\label{sec:defs}
The \textit{utilitarian social welfare} $\UW$ and \textit{Nash welfare} $\NW$ of an allocation are respectively defined as the sum and geometric mean\footnote{The most rigorous definition of $\NW$  carefully accommodates instances that admit no allocations with strictly positive realized valuations for all agents, but no such instance arises in our study.} of the realized valuations of all agents:

$\UW (A) \coloneq \sum_{i \in N} v_i(A_i)$; \hfill $\NW (A) \coloneq \left(\prod_{i \in N} v_i(A_i)\right)^{1/n}$.

In an allocation $A$, agent $i$ is said to \textit{envy} agent $j$ if $v_i(A_i) < v_i(A_j)$, otherwise agent $i$ is \textit{envy-free} of agent $j$. The following are two popular relaxations of envy-freeness for integral allocations that will inform our performance metrics. Agent $i$ is \textit{envy-free} of agent $j$ \textit{up to one item} if $i$ envies $j$ but there is an item $o \in A_j$ such that $v_i(A_i) \ge v_i(A_j \setminus \{o\})$. 
Under additive valuations, this criterion reduces to the following inequality, which places an upper bound on the envy of $i$ toward $j$:
\begin{align}
    \sum_{k=1}^m V_{o_k,i} A_{o_k,j} - \sum_{k=1}^m V_{o_k,i} A_{o_k,i} \le \max_{k: A_{o_k,j}=1} V_{o_k,i} \label{def:EF1}
\end{align}
The derivation is in App.~\ref{app:additive_EF1_EFX}. An allocation is EF1 
if, for every ordered pair of distinct agents $(i,j)$, either $i$ does not envy $j$ or $i$ is envy-free of $j$ up to one item. 

Finally, note that any tie arising among agents, ordered agent-pairs or items in any procedure we describe in the rest of the paper could in general be broken arbitrarily but consistently; we always use lexicographical tie-breaking.





\begin{figure}[t]
  \centering
  \includegraphics[width=0.8\linewidth]{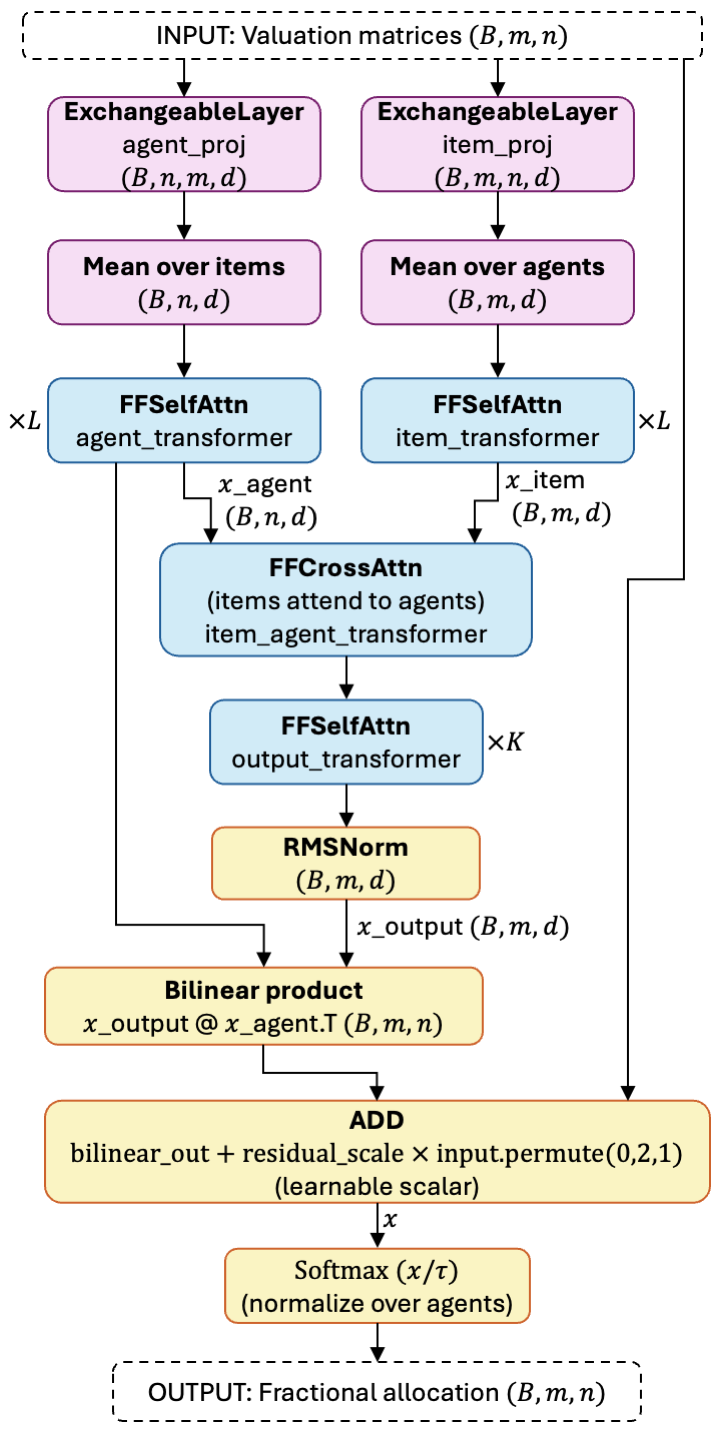}
  \caption{\FF architecture: Agents and items are embedded separately, interact via cross-attention, and produce per-item assignment distributions through a temperature-scaled softmax.}
  \label{fig:fairformer}
\end{figure}

\section{Transformer-based allocator}
We will now describe our two-stage approach to discrete fair division. The first and major component is \FF, a transformer trained to maximize Nash welfare (Section~\ref{sec:fairformer}). This produces fractional allocations, which we discretize and then subject to a post-processing routine \EFoneR (Section~\ref{sec:ef1repair}). 

\subsection{\FF}
\label{sec:fairformer}


Fig.~\ref{fig:fairformer} depicts a schematic for \FF. Given $V\in\mathbb{R}_{\ge 0}^{m\times n}$, \FF implements a mapping
$f_\theta: \mathbb{R}_{\ge 0}^{m\times n} \to (\Delta^n)^m$ in two steps, where $\theta$ represents the parameters of the neural network and $\Delta^n$ is the probability simplex representing fractional distributions of an item. 

First, we produce an \emph{item--agent score matrix} $S_\theta(V) \in \mathbb{R}^{m\times n}$,
which we interpret as a \emph{modified utility matrix}. It has the same shape as $V$
but is augmented by learned structure. The generation of $S_\theta(V)$ is guided by the design principle of permutation-equivariance over both agents and items, i.e., permuting agent indices (columns of $V$) or item indices (rows of $V$) should only induce an analogous permutation in the corresponding dimension of $S_\theta(V)$. These properties align the architecture with the symmetry of fair division instances,
yielding strong generalization to unseen $m$ and $n$ distributions. 
The architectural details are as follows.

\paragraph{Tokenization and projections.}
We represent agents and items as disjoint sets of tokens obtained by projecting the valuation matrix
along opposite axes:
\[
H_0=\phi_A(V^\top)\in\mathbb{R}^{n\times d}, \qquad
I_0=\phi_I(V)\in\mathbb{R}^{m\times d},
\]
where $d$ is the model dimension, and $\phi_A$, $\phi_I$ are learned \textit{exchangeable layers} \cite{hartford2018deep}. 

\paragraph{Building blocks.} We employ two types of attention layer, $\SAttn(X)$ for self-attention and $\CAttn(X_q, X_{kv})$ for cross-attention, where $X$ represents an arbitrary input with model dimension $d$; the subscripts $q$ and $kv$ represent query and key-value pair in the standard terminology of the attention literature \cite{vaswani2017attention}, corresponding to items and agents respectively in this context.
Each layer is a residual pre-norm \cite{xiong2020layer} block with \textit{multi-head attention} (MHA) \cite{vaswani2017attention}
and a gated linear unit (GLU) \cite{shazeer2020glu}. In other words, the output of each attention layer takes the form $Y = \hat{X} + \mathrm{GLU}(\norm(\hat{X}))$, where $\hat{X} = X + \mathrm{MHA}(\norm(X))$ for self-attention and $\hat{X} = X_q + \mathrm{MHA}(\mathrm{Norm}(X_q),\,\mathrm{Norm}(X_{kv}))$ for cross-attention. The specific $\norm$ function we adopt is the \textsc{RMSNorm} \cite{zhang2019root}. For our GLU, we use SiLU \cite{hendrycks2023gaussian} activations with intermediate width
$\tfrac{8}{3}d$. Attention is unmasked and non-causal.





\paragraph{Two-tower encoder and fusion.}
Agents and items are first encoded independently via $L$ self-attention layers each: 
\[
H_{\ell + 1}=\SAttn(H_\ell), \qquad
I_{\ell + 1}=\SAttn(I_\ell).
\]
for $\ell=0,\dots,L-1$. Cross-attention then fuses the two towers by allowing items to query agent representations:
\[
Z_0=\CAttn(I_L,H_L)\in\mathbb{R}^{m\times d}.
\]
This operation contextualizes each item embedding with respect to agent valuations,
analogous to computing a context-aware “responsibility” of agents for each item.

\paragraph{Allocation refinement.}
To enable global consistency among items (e.g., discouraging duplication or conflicting assignments),
we stack $K$ additional item-wise self-attention layers $Z_{\ell + 1}=\SAttn(Z_{\ell})$, $\ell = 1,\dots,K$, followed by normalization to obtain $\tilde{Z}\in\mathbb{R}^{m\times d}$.

\paragraph{Bilinear output + residual.}
For final item embeddings $\tilde Z\in\mathbb{R}^{m\times d}$ and final agent embeddings 
$H_L\in\mathbb{R}^{n\times d}$ produced by the two-tower encoder,
the model computes a bilinear compatibility matrix
\[
D ~=~ \tilde{Z} H_L^\top \in \mathbb{R}^{m\times n}.
\]
Crucially, we then add a learnable residual connection from the raw valuations to obtain the modified utility matrix:
\[
S_\theta(V) ~=~ D ~+~ \alpha V,
\]
where $\alpha\in\mathbb{R}$ is a learned scalar, the residual scale parameter. The key idea is that
$S_\theta(V)$ preserves the direct “give item to high-valuer” signal via $V$ while allowing the
network to apply structured, fairness-relevant corrections via $D$.

Finally, an allocation matrix is generated by applying a \textsc{Softmax} layer  to the modified utility matrix. The \textsc{Softmax} function operates across agents (columns) independently for each item (row) and is controlled by temperature $\tau$ which is annealed during training (see Section~\ref{sec:training}):
\[
A = \mathrm{\textsc{Softmax} }\!\left(S_\theta(V)/\tau\right) \in (\Delta^n)^m.
\]



\subsubsection{Training}
\label{sec:training}

\FF is trained end-to-end using the differentiable surrogate of
Nash welfare \cite{caragiannis2019unreasonable} as the objective on batches of valuation matrix samples of size $B$. 
For a batch of valuation matrices $\{V^r\}_{r=1}^B$
and corresponding predicted fractional allocations, i.e., model outputs $\{A^r\}_{r=1}^B$,
we seek to maximize via gradient descent the expected log-Nash welfare:
\[
\mathcal{L}_\text{NW}(\theta)
~=~
-\frac{1}{B}\sum_{r=1}^{B}
\frac{1}{n}\sum_{i=1}^{n}
\log\!\Big(
  \sum_{k=1}^m V^r_{o_k,i} A^r_{o_k,i}
\Big),
\]
This objective directly encourages allocations in which all agents
receive strictly positive and balanced values, penalizing degenerate
solutions that favor a strict subset of agents.

We anneal the \textsc{Softmax} temperature $\tau$ from an initial
$\tau_0 > 0$ to a final $\tau_T > 0$ over the training horizon $T$ to transition smoothly
from fractional per-item distributions over agents (which can be viewed as exploration) to nearly one-hot item-to-agent assignments. The permutation-equivariant structure of the architecture ensures that learning is shared across agents and items.


\subsection{\EFoneR}
\label{sec:ef1repair}
Since \FF produces a fractional allocation $A$ in general, we convert $A$ to a discrete allocation $\hat{A}\in\{0,1\}^{m\times n}$ by a deterministic row-wise $\arg\max$ operator at inference/evaluation time, i.e., assigning item $o_k$ to $\arg\max_{i \in N} A_{o_k,i}$. However, there is no guarantee that the output is EF1 after discretization.  
Hence, we invoke a post-processing routine \EFoneR (Algorithm~\ref{alg:ef1repair}) that enforces the EF1 property while improving Nash welfare. 



Given an (integral) allocation $\hat{A}$, we iterate through all ordered pairs of distinct agents $(i,j)$ in each pass of \EFoneR to check if each pair violates the EF1 condition~\eqref{def:EF1}, i.e., if 
\begin{align}
    \sum_{k=1}^m V_{o_k,i} \hat{A}_{o_k,j} - \sum_{k=1}^m V_{o_k,i} \hat{A}_{o_k,i} > \max_{k: \hat{A}_{o_k,j}=1} V_{o_k,i} \label{ineq:EF1_viol}
\end{align}
If such a violation occurs under additive valuations, the proof of Theorem 3.2 of \citet{caragiannis2019unreasonable} implies that there is an item $o\in \hat{A}_j$ that we can transfer from $j$ to $i$ such that $\NW(\hat{A}')>\NW(\hat{A})$ where $\hat{A}'$ is the updated allocation after the transfer.
Among all such items, we transfer the one maximizing the log-Nash welfare improvement, which is equivalent to maximizing the following function over items in $\hat{A}_j$:
\begin{align}
    \Delta^{\hat{A}}_{i,j}(o)
\coloneq
\log(v_i(\hat{A})+V_{o, i}) + \log(v_j(\hat{A}) - V_{o, j}) \label{def:NWimprov}
\end{align}
Since the maximum Nash welfare is finite for any fair division instance with real valuations and the Nash welfare strictly increases with every pass that transfers at least one item, this procedure must terminate. We cap the number of passes at $P$. Since each pass checks all ordered pairs and scans all items in the envied bundle, the runtime of \EFoneR is $O(P  n^2 m)$.
We provide both a reference implementation and a Numba-accelerated implementation.

\begin{algorithm}[t]
\caption{\EFoneR (post-processing)}
\label{alg:ef1repair}
\begin{algorithmic}[1]
\Require Allocation $\hat{A}\in\{0,1\}^{m\times n}$; Valuations $V\in\mathbb{R}_{\ge0}^{m\times n}$; Maximum number of passes $P$
\State \textbf{Initialize:} $p \gets 1$; $changed \gets \mathrm{\T}$
\While{$p\le P$  \textbf{and} $changed$}
  \State $p \gets p+1$; $changed \gets \mathrm{\F}$
  \For{all ordered pairs $(i,j)$ with $i\neq j$}
    \If{$\hat{A}_j\neq\emptyset$ \textbf{and}  Inequality~\eqref{ineq:EF1_viol} holds}
      \State Choose $o^\star\in \arg\max_{o \in \hat{A}_j} \Delta^{\hat{A}}_{i,j}(o)$, see \eqref{def:NWimprov}
      \State \textbf{Transfer}: $\hat{A}_{o^\star,j} \gets 0$; $\hat{A}_{o^\star,i} \gets 1$
      \State $changed \gets \mathrm{\T}$
    \EndIf
  \EndFor
\EndWhile\\
\Return $\hat{A}$
\end{algorithmic}
\end{algorithm}

\section{Evaluation experiments}\label{sec:experiments}
We will now report the salient aspects and results of our evaluation experiments. App.~\ref{app:additional_results} provides some omitted details. We conducted all training and evaluation on an RTX 4090. 

\subsection{Comparison benchmarks}
We contrast our approach with three indivisible goods allocation algorithms from the literature. 

\paragraph{Round-robin protocol (RR) \cite{caragiannis2019unreasonable}.} 
Agents $1,\dots,n$ begin with empty bundles and take turns in each round in the same order; at its turn, we assign to an agent its most-valued remaining item until we run out of items. RR guarantees an ex-post EF1 allocation for additive valuations and obviously runs in polynomial time.

\paragraph{Envy Cycle Elimination (ECE) \cite{LiptonMaMo04}.}
Again, the agents start with empty bundles. The algorithm maintains an \textit{envy graph} on agents with a directed edge from $i$ to $j$ whenever $i$ envies $j$. In each iteration, 
we assign the next remaining item to a currently unenvied agent (i.e., one which no incoming edge); if this creates a directed cycle, we eliminate it by transferring bundles from envied to envious agent on each such cycle (in the reverse direction, increasing the realized valuation of each agent on that cycle), until we get at least one unenvied agent (with indegree $0$). This algorithm guarantees an ex-post EF1 allocation for monotone valuations and is also polynomial-time.


\paragraph{MaxUtil + EF1-Repair (MU+EF1).}
This baseline is similar to the \textit{envy-induced transfers} procedure of \citet{benabbou2021finding} who applied it to a non-additive valuation class. First, for a given instance, we compute a utilitarian-optimal allocation with maximum \UW by iterating over items and assigning each item to an agent with the highest value for it; this exploits the fact that utilitarian welfare decomposes across items under additivity. Finally, we apply \EFoneR from Section~\ref{sec:ef1repair} to the above allocation.
This baseline is intentionally strong on utilitarian welfare before repair, and exposes the welfare-fairness trade-off
induced by the repair dynamics.

\subsection{Data sets and training configurations}

As our main evaluation data set, we generated a cohort of $R=10^4$ discrete fair division instances $V^r$, $r=1,2,\dots,R$ with $n$ varying within $\{10,11,\dots,60\}$, $m$ from $n$ to $60$, and all item-agent valuations being sampled independently over agents, items, and instances from $\unif[0,1]$.

On this data set, we test three \FF models that are trained on three separate sets of valuation matrices with each entry $V_{o_k,i}\sim_{\mathrm{i.i.d.}}  \unif[0, 1]$: (1) \textbf{small data regime ($10{\times}20$)}, exclusively trained on instances with $n=10, m=20$; (2) \textbf{medium data regime ($30{\times}60$)}, exclusively trained on instances with $n=30, m=60$; (3) \textbf{multi-config}, trained on many $(n, m)$ configurations (but always $m \geq n)$.
We report the hyperparameters for all three models are in appendix \ref{app:config}.

Additionally, we use for evaluation two non-uniform valuation distributions that capture agent preferences in certain realistic scenarios. 
Table~\ref{tab:distributions}  in App.~\ref{app:val_dist} summarizes the key characteristics of each distribution.                                      \paragraph{(Normalized) Pareto-Distributed Valuations.}                              For some resources such as computing jobs with varying resource demands, advertisement slot auctions, and courses with varying popularity levels, item values can follow heavy-tailed distributions, i.e., a few items are highly valued while most have moderate value. We model this by first sampling for each instance ``raw" valuations $V_{o_k,i}^{\text{raw}}$ i.i.d. from a Pareto distribution with parameter $\alpha$, i.e., with cumulative distribution function $F(x) = 1 - x^{-\alpha}$. Specifically, we use use $\alpha = 3.0$, which produces moderate heavy-tailedness.


   These raw valuations are then min-max normalized per instance:
  \begin{equation*}
      V_{o_k.i} = \frac{V_{o_k,i}^{\text{raw}} - \min_{o,j} V_{o,j}^{\text{raw}}}{\max_{o,j} V_{o,j}^{\text{raw}} - \min_{o,j} V_{o,j}^{\text{raw}}}
  \end{equation*}
  This normalization ensures valuations lie in $[0,1]$ while preserving the relative structure of heavy-tailed preferences.

  \paragraph{Correlated Valuations.}
  In some domains, agents have correlated preferences due to objective quality differences between items, 
  creating
  competition for high-quality items. To capture this, we give each item $o$ in an instance an intrinsic quality $\beta_o \sim_{\mathrm{i.i.d.}}  \unif[0, 1]$ shared across all agents, and take its convex combination with agent-specific idiosyncratic noise $\varepsilon_{o,i}\sim_{\mathrm{i.i.d.}}  \unif[0, 1]$ to obtain the item-agent valuation
  $V_{o,i} = \lambda \beta_o + (1 - \lambda) \varepsilon_{o,i}$.
  We use $\lambda = 0.5$, equally balancing common value and bias. 

For each of these two non-uniform distributions, we generate data sets with multiple $(n,m)$ configurations and employ a \FF model trained on the corresponding small data regime; we call them FF-P and FF-C respectively (see Tables~\ref{tab:pareto-results} and \ref{tab:correlated-results}).


\subsection{Metrics}
We subjected each instance to every allocator described above, and recorded the wall time as well as the number of iterations (of each applicable component) needed to generate an allocation. In particular for our proposed approach, we applied one forward pass of the relevant trained \FF model, discretized the output via row-wise $\arg\max$, and finally refined it with \EFoneR.

For each instance $V^r$, we pre-computed the maximum utilitarian and Nash welfares over all allocations, $\max_{A}\UW(A)$ and $\max_{A}\NW(A)$. For the first, we adopted the same method as for MaxUtil+EF1; 
for the second, we constructed and solved the following mixed-integer convex program using Gurobi \cite{gurobi}: 
\begin{align*}
\max & \sum_{i=1}^n \ln (\sum_{k=1}^m V^r_{o_k, i}A_{o_k,i}) \\
\text{s.t.} & \sum_{i=1}^n A_{o_k,i} = 1 \forall k =1,\dots, m,\\
& A_{o_k,i} \in \{0,1\} \forall i=1,\dots,n, \forall k =1,\dots, m
\end{align*}



For each discrete allocation $\hat{A}$ produced by a method, we compute $\UW(\hat{A})$ and $\NW(\hat{A})$
and normalize them by the instance-specific optima $\max_A \UW(A)$ and $\max_A \NW(A)$.
We report the mean normalized welfare across all $(n,m)$ configurations in \cref{tab:welfare_summary},
and analyze performance as a function of the items-per-agent ratio $m/n$ (an indicator for scarcity) in \cref{fig:welfare_ratio}.
Table~\ref{tab:welfare_summary} summarizes mean$\pm$std and min--max over the full evaluation suite. For non-uniform distributions, we report \NW only in Tables~\ref{tab:pareto-results} and \ref{tab:correlated-results}.


\subsection{Results}
\label{sec:results}


\begin{figure}[t]
  \centering
  \includegraphics[width=\columnwidth]{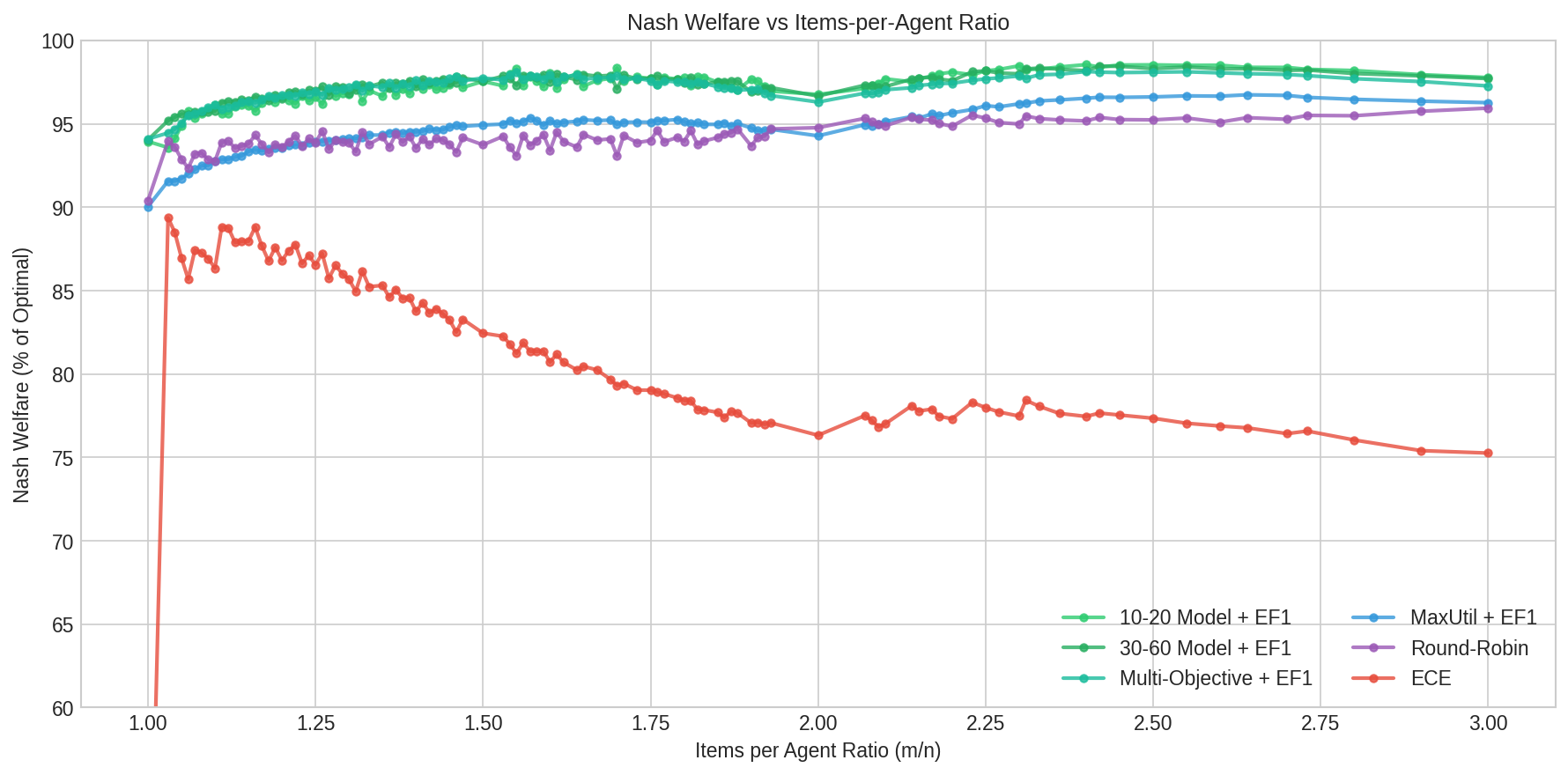}
  \caption{Nash welfare as a function of the items-per-agent ratio $m/n$. Learned models remain near-optimal across scarcity regimes, while classical EF1 procedures exhibit larger degradation, especially ECE.}
  \label{fig:welfare_ratio}
\end{figure}

We evaluate three trained \FF variants (trained on $10{\times}20$, trained on $30{\times}60$, and a multi-configuration model) against Round-Robin (RR), Envy Cycle Elimination (ECE), and a strong utilitarian baseline (MaxUtil) followed by the same EF1 repair used for \FF.
For each instance, we report utilitarian and Nash welfare \emph{as a percentage of the instance-wise optimum} (Section~\ref{sec:prelim}), so values can be interpreted as approximation ratios.
\EFoneR produced an EF1 allocation for all model and MaxUtil outputs within three passes (\cref{tab:ef1repair_passes_mean} in App.~\ref{app:additional_results}).

\begin{table*}[t]
\centering
\small
\begin{tabular}{lccccc}
\toprule
Method & Nash (\% opt) & Nash min--max & Util (\% opt) & Util min--max \\
\midrule
\FF\ ($10{\times}20$) + EF1-Repair     & $96.62 \pm 1.42$ & $[91.80, 98.61]$ & $95.57 \pm 1.70$ & $[92.02, 98.67]$ \\
\FF\ ($30{\times}60$) + EF1-Repair     & $96.81 \pm 1.18$ & $[93.29, 98.46]$ & $95.76 \pm 1.62$ & $[90.48, 98.37]$ \\
\FF\ (Multi-Config) + EF1-Repair          & $96.70 \pm 1.20$ & $[92.98, 98.15]$ & $95.88 \pm 1.60$ & $[92.70, 98.78]$ \\
MaxUtil + EF1-Repair                   & $93.89 \pm 1.71$ & $[89.28, 96.74]$ & $93.56 \pm 2.09$ & $[89.57, 98.07]$ \\
Round-Robin                     & $93.70 \pm 1.44$ & $[87.18, 95.95]$ & $92.26 \pm 1.31$ & $[87.55, 94.66]$ \\
Envy Cycle Elimination         & $79.32 \pm 12.75$& $[39.15, 90.43]$ & $78.96 \pm 9.14$ & $[51.43, 89.52]$ \\
\bottomrule
\end{tabular}
\caption{Summary statistics of welfare ratios (mean $\pm$ std; min--max) over all evaluated problem configurations.}
\label{tab:welfare_summary}
\end{table*}

\paragraph{Near-optimal welfare under EF1.}
\cref{tab:welfare_summary} summarizes performance averaged across all evaluated $(n,m)$ configurations.
All three \FF models achieve \textbf{$\approx 96.6$--$96.8\%$} of optimal Nash welfare and \textbf{$\approx 95.6$--$95.9\%$} of optimal utilitarian welfare (Table~\ref{tab:welfare_summary}).
The best Nash mean is attained by the $30{\times}60$-trained model ($96.81\%$), while the best utilitarian mean is attained by the multi-configuration model ($95.88\%$).
In contrast, MaxUtil+EF1-Repair attains $93.89\%$ Nash and $93.56\%$ utilitarian welfare, and RR attains $93.70\%$ Nash and $92.26\%$ utilitarian welfare.
ECE is substantially worse on both objectives ($79.32\%$ Nash, $78.96\%$ utilitarian) and exhibits much larger dispersion.

\paragraph{Consistency and failure modes.}
Beyond means, \FF is markedly more \emph{stable} across configurations.
Across the full evaluation suite, the worst observed Nash ratios for the learned models remain above $91.80\%$ (and above $92.98\%$ for the multi-configuration model), whereas ECE drops as low as $39.15\%$ of optimal Nash welfare (Table~\ref{tab:welfare_summary}).
A similar pattern holds for utilitarian welfare: learned models remain above $90.48\%$ of the per-instance utilitarian optimum, while ECE reaches minima of $51.43\%$.
This indicates that enforcing EF1 alone is insufficient for good efficiency: the specific allocation rule matters substantially.

\paragraph{Dependence on scarcity ($m/n$).}
\cref{fig:welfare_ratio} breaks Nash welfare down by the items-per-agent ratio $m/n$.
All \FF variants remain close to optimal across the full tested range ($m/n \in [1,3]$), with only mild degradation at the extremes.
RR and MaxUtil+EF1-Repair are consistently below the learned models across the range.
ECE degrades sharply as $m/n$ increases, explaining its low average performance and high variance.

\paragraph{Significance.}
Pairwise tests over the pooled evaluation instances show that the welfare gains of \FF over baselines are not only large but also statistically significant.
For Nash welfare, the multi-configuration \FF improves over RR by $3.00$ percentage points (95\% CI $[2.83, 3.18]$) and over MaxUtil+EF1-Repair by $2.81$ points (95\% CI $[2.72, 2.91]$), with extremely small $p$-values under both $t$-tests and Wilcoxon tests, and large standardized effects (e.g., Cohen's $d \ge 2.25$ vs.\ RR and $d \ge 3.92$ vs.\ MaxUtil+EF1-Repair).
Analogous improvements hold for utilitarian welfare ($+3.62$ points vs.\ RR and $+2.32$ points vs.\ MaxUtil+EF1-Repair).
Differences among the three \FF training regimes are measurable but very small in magnitude (sub-$0.3$ points), indicating that all three variants learn high-quality allocation rules on this distribution.

\begin{table*}[t]
  \centering
  \small
  \begin{tabular}{lrrrrr}
    \toprule
    Algorithm & $10{\times}10$ & $10{\times}15$ & $10{\times}20$ & $20{\times}40$ & $20{\times}60$ \\
    \midrule
    \FF\ (10--20)   & 414   & 946   & 1{,}088 & 1{,}402   & 2{,}825 \\
    \FF\ (30--60)   & 430   & 528   & 665     & 1{,}244   & 2{,}196 \\
    \FF\ (Multi)    & 447   & 507   & 691     & 1{,}752   & 2{,}697 \\
    Round-Robin     & 87    & 163   & 170     & 365       & 559 \\
    ECE             & 311   & 805   & 1{,}107 & 8{,}815   & 14{,}834 \\
    MaxUtil         & 0.3   & 0.4   & 0.5     & 1.5       & 2.4 \\
    \bottomrule
  \end{tabular}
  \caption{Mean per-instance runtime (in $\mu$s) for representative problem sizes.}
  \label{tab:runtime}
\end{table*}


\paragraph{Runtime.}
\cref{tab:runtime} reports mean per-instance runtime (in $\mu$s) for representative problem sizes.
A single \FF forward pass is sub-millisecond on small instances (e.g., $10{\times}10$: 414--447~$\mu$s across checkpoints)
and remains within a few milliseconds at $20{\times}60$ (2,196--2,825~$\mu$s).
Round-Robin is faster (87--559~$\mu$s) but yields systematically lower welfare in our experiments,
while Envy-Cycle Elimination scales sharply with problem size (up to 14,834~$\mu$s at $20{\times}60$).
MaxUtil is faster but does not enforce fairness constraints and remains roughly comparable to RR after EF1-repair.

  \paragraph{Non-uniform valuation distributions.}
  Standard baselines like Round Robin perform well on uniform valuations but can struggle when valuations have structure. Pareto distributions particularly challenge methods like Envy-Cycle Elimination (ECE) because the geometric mean in Nash welfare is sensitive to agents receiving low-value bundles---a common outcome when a few items dominate total value. Our experiments show that
  distribution-specific training yields models that achieve $97.9$--$99.1\%$ of optimal Nash welfare on Pareto valuations and $96.5$--$99.4\%$ on correlated valuations, significantly outperforming distribution-agnostic baselines (Tables~\ref{tab:pareto-results} and \ref{tab:correlated-results})).

 \begin{table}[t]
  \small
  \centering
  \caption{Nash welfare as percentage of optimal on Pareto($\alpha=3.0$) valuations.} 
  \label{tab:pareto-results}
  \begin{tabular}{lcccc}
  \toprule
  \textbf{Size} & \textbf{FF-P+EF1} & \textbf{MU+EF1} & \textbf{RR} & \textbf{ECE} \\
  \midrule
  $3 \times 6$   & \textbf{97.9\%} & 93.9\% & 93.1\% & 73.4\% \\
  $3 \times 9$   & \textbf{98.0\%} & 94.3\% & 93.7\% & 70.9\% \\
  $3 \times 12$  & \textbf{98.9\%} & 95.3\% & 95.4\% & 67.6\% \\
  $5 \times 10$  & \textbf{98.0\%} & 91.9\% & 88.7\% & 66.4\% \\
  $5 \times 15$  & \textbf{98.8\%} & 93.3\% & 92.7\% & 62.5\% \\
  $5 \times 20$  & \textbf{99.1\%} & 94.0\% & 94.2\% & 60.1\% \\
  $10 \times 20$ & \textbf{98.3\%} & 87.9\% & 87.6\% & 60.3\% \\
  $10 \times 30$ & \textbf{98.3\%} & 91.0\% & 92.3\% & 55.4\% \\
  \bottomrule
  \end{tabular}
\end{table}

 \begin{table}[t]
  \small
  \centering
  \caption{Nash welfare as percentage of optimal on Correlated($\lambda=0.5$) valuations.}
  \label{tab:correlated-results}
  \begin{tabular}{lcccc}
  \toprule
  \textbf{Size} & \textbf{FF-C+EF1} & \textbf{MU+EF1} & \textbf{RR} & \textbf{ECE} \\
  \midrule
  $3 \times 6$   & \textbf{98.9\%} & 98.1\% & 94.0\% & 88.5\% \\
  $3 \times 9$   & \textbf{99.1\%} & 98.5\% & 95.9\% & 87.3\% \\
  $3 \times 12$  & \textbf{99.4\%} & 98.8\% & 96.9\% & 86.2\% \\
  $5 \times 10$  & \textbf{98.5\%} & 97.0\% & 93.4\% & 85.1\% \\
  $5 \times 15$  & \textbf{98.9\%} & 98.3\% & 95.2\% & 84.2\% \\
  $5 \times 20$  & \textbf{99.0\%} & 98.3\% & 96.4\% & 83.9\% \\
  $10 \times 20$ & \textbf{97.5\%} & 96.1\% & 93.5\% & 83.1\% \\
  $10 \times 30$ & \textbf{98.3\%} & 97.8\% & 95.2\% & 82.8\% \\
  $15 \times 30$ & \textbf{97.2\%} & 95.8\% & 93.4\% & 83.5\% \\
  $20 \times 40$ & \textbf{96.5\%} & 95.2\% & 94.0\% & 83.7\% \\
  \bottomrule
  \end{tabular}
\end{table}

\section{Discussion}
\label{sec:discussion}
We proposed and evaluated \FF, a compact, symmetry-preserving, transformer-based network that produces high-quality indivisible goods allocations with minimal (but theoretically well-founded) post-processing.

\paragraph{Main takeaway: high efficiency \emph{and} EF1 at scale.}
Across a broad range of $(n,m)$, all three \FF variants consistently produce EF1 allocations (via repair) whose welfare is close to the per-instance optimum.
In aggregate, \FF attains $\approx 96$--$97\%$ of optimal Nash welfare and $\approx 95$--$96\%$ of optimal utilitarian welfare, substantially outperforming classical EF1 baselines (RR, ECE) and a strong ``greedy then repair'' baseline (MU+EF1).
Since the (integral) maximum Nash welfare allocation is EF1 under additive valuations, the Nash ratios here can be read directly as approximation quality to an \emph{EF1-feasible} optimum.

\paragraph{Why MaxUtil+repair is not enough.}
A natural approach to efficiency is to start from the utilitarian optimum and repair fairness violations.
Empirically, this strategy underperforms \FF on \emph{both} Nash and utilitarian welfare after repair.
This indicates that the repair dynamics can induce avoidable welfare loss when starting from a highly unbalanced allocation, and that learning a policy which anticipates EF1 constraints (even if enforced post hoc) yields systematically better outcomes.

\paragraph{ECE exhibits large variance and can catastrophically fail on welfare.}
ECE provides an existence/constructive EF1 guarantee, but our results show that its welfare can be extremely poor on the same i.i.d.\ additive instances where \FF remains stable.
The heavy tail (very low minima and large standard deviation) suggests that ECE's local graph-based dynamics can allocate many items in ways that are fairness-preserving yet globally inefficient, particularly as $m/n$ grows.

\paragraph{Training distribution effects are small on average, but matter for a single checkpoint.}
The three \FF training regimes are tightly clustered, with sub-$0.3$ point differences in mean welfare.
For Nash welfare, the $30{\times}60$-trained model is marginally best; for utilitarian welfare, the multi-configuration model is marginally best.
Given the small gaps, the multi-configuration model is the most robust ``single model'' choice when both welfare metrics and cross-size generalization are desired.

\paragraph{Future work.} A natural direction is to extend/modify the transformer architecture so that it can handle valuation classes beyond additive, a domain with few theoretical guarantees. Although arbitrary monotone valuations may be a tall order, certain submodular valuations such as budget-limited or capacitated additive may be amenable.





\FloatBarrier
\newpage
\bibliographystyle{ACM-Reference-Format} 
\bibliography{fairformerRefs}


\appendix
\onecolumn
\clearpage

\section*{Technical appendices to Submission XXX}

\section{Alternative formulations of EF1 and EFX conditions for additive valuations}\label{app:additive_EF1_EFX}

\begin{proposition}
    For a discrete fair instance with additive valuations, let $V$ denote the valuation matrix and $A$ denote an arbitrary allocation matrix. Suppose agent $i$ envies agent $j$ in allocation $A$. Then agent $i$ is envy-free of agent $j$ up to one item (EF1) if and only if
    \begin{align}
        \sum_{k=1}^m V_{o_k,i} A_{o_k,j} - \sum_{k=1}^m V_{o_k,i} A_{o_k,i} \le \max_{k: A_{o_k,j}=1} V_{o_k,i}.\label{app:EF1}
    \end{align}
\end{proposition}

\begin{proof}
    Since $i$ envies $j$, we must have $v_i(A_j) > v_i(A_i) \ge 0$, hence $A_j$ is non-empty, i.e., $A_{o_k,j}=1$ for at least one $k \in \{1,2,\dots,m\}$. Hence, the maximum in Inequality~\eqref{app:EF1} is well-defined.
    
    
    It suffices to show that the following general expression for the EF1 condition is equivalent to Inequality~\eqref{app:EF1} in our setting:
    \begin{align}
        \exists o \in A_j . v_i(A_i) \ge v_i(A_j \setminus \{o\}). \label{ineq:EF1_gen}
    \end{align}
  To show that \eqref{ineq:EF1_gen} $\Rightarrow$ \eqref{app:EF1}, we pick an item $o_l \in A_j$ that satisfies the inequality in \eqref{ineq:EF1_gen}. Then, under additive valuations, the LHS of this inequality is $\sum_{k=1}^m V_{o_k,i} A_{o_k,i}$ and the RHS is
  \begin{align*}
       v_i(A_j \setminus \{o_l\}) = \sum_{o\in A_j\setminus\{o_l\}} v_i(o) = \sum_{o\in A_j} v_i(o) - v_i(o_l) = v_i(A_j) - v_i(o_l) =  \sum_{k=1}^m V_{o_k,i} A_{o_k,j} - V_{o_l,i}.
  \end{align*}
Plugging these expressions back and rearranging, we get
\[ \sum_{k=1}^m V_{o_k,i} A_{o_k,j} - \sum_{k=1}^m V_{o_k,i} A_{o_k,i} \le V_{o_l,i} \le \max_{k: A_{o_k,j}=1} V_{o_k,i},\]
since $A_{o_l,j}=1$ by our definition of $o_l$.

To show that \eqref{app:EF1} $\Rightarrow$ \eqref{ineq:EF1_gen}, we simply pick an item $o_l$ such that $l \in \arg\max_{k: A_{o_k,j}=1} V_{o_k,i}$. Obviously, $o_l \in A_j$ since $A_{o_l,j}=1$ by definition; we then mirror the above steps to get $v_i(A_i) \ge v_i(A_j \setminus \{o_l\})$.
\end{proof}

\section{\FF configurations}
\label{app:config}

\begin{table}[h]
\centering
\caption{Model configurations used in experiments.}
\label{tab:model-configs}
\begin{tabular}{lccccc}
\hline
\textbf{Model} &
$\mathbf{d_{\text{model}}}$ &
$\mathbf{\#\;heads}$ &
$\mathbf{\#\;encoder\;layers}$ &
$\mathbf{\#\;output\;layers}$ &
\textbf{Dropout} \\
\hline
10$\times$20 Model        & 256 & 8 & 1 & 2 & 0.0   \\
30$\times$60 Model        & 128 & 8 & 3 & 2 & 0.099 \\
Multi-Configuration Model    & 256 & 8 & 1 & 2 & 0.0   \\
\hline
\end{tabular}
\end{table}

\FloatBarrier

\section{Valuation distributions used for evaluation }\label{app:val_dist}
 \begin{table}[h]
  \small
  \centering
  \caption{Valuation distributions and their characteristics.}
  \label{tab:distributions}
  \setlength{\tabcolsep}{4pt}
  \begin{tabular}{@{}lll@{}}
  \toprule
  \textbf{Distribution} & \textbf{Key Challenge} & \textbf{Real-World Analog} \\
  \midrule
  Uniform & Baseline (no structure) & Random preferences \\
  Pareto & Heavy tails, high-value outliers & Auctions, resource allocation \\
  Correlated & Competition for shared items & Course allocation, job markets \\
  \bottomrule
  \end{tabular}
  \end{table}
\FloatBarrier

\section{Additional experimental details and results}
\label{app:additional_results}

This appendix provides complementary views of the welfare results in the main text.
All plots use the same normalization as in Section~\ref{sec:results}: for each instance, we report Nash welfare and utilitarian welfare as a percentage of the instance-wise optimum, and then aggregate across instances and/or problem configurations.

\paragraph{Distributional stability across configurations.}
Figure~\ref{fig:welfare_boxplots} visualizes the distribution of welfare ratios (across all evaluated $(n,m)$ configurations).
The learned \FF variants exhibit tightly concentrated boxplots, indicating that their strong mean performance is not driven by a small subset of easy instances.
In contrast, Envy Cycle Elimination (ECE) exhibits heavy left tails and wide interquartile ranges, consistent with the large variance and low minima reported in Table~\ref{tab:welfare_summary}.

\paragraph{Scaling with total problem size.}
Figure~\ref{fig:nash_vs_size} plots normalized Nash welfare against the total problem size $(n+m)$.
Across the evaluated range, the \FF variants remain near-flat, suggesting that the symmetry-aware architecture and Nash-welfare training objective transfer smoothly across scale.
ECE shows pronounced instability as size increases, reflecting that EF1 feasibility alone does not prevent severe welfare degradation.

\paragraph{Cross-size generalization within \FF.}
Figure~\ref{fig:heatmaps_models_only} shows heatmaps of normalized Nash welfare over the $(n,m)$ grid for each training regime.
Although each model is trained on a restricted set of sizes (single-size training for $10{\times}20$ and $30{\times}60$, versus multi-configuration training),
all three achieve strong performance well outside their training sizes.
The heatmaps provide a fine-grained view of where each training regime is marginally strongest, complementing the pooled summary statistics in the main paper.

\paragraph{Statistical testing details.}
Figure~\ref{fig:significance_tables} summarizes paired comparisons of normalized welfare ratios between methods on the same evaluation instances.
We report mean paired differences with 95\% confidence intervals, standardized effect sizes (Cohen’s $d$), and $p$-values from both paired $t$-tests and Wilcoxon signed-rank tests.
Because the evaluation set is large, very small differences (including among \FF variants) can be statistically significant;
the practically meaningful comparisons are the consistently larger gaps between \FF and the classical EF1 baselines.

\begin{figure*}[t]
  \centering
  \includegraphics[width=\textwidth]{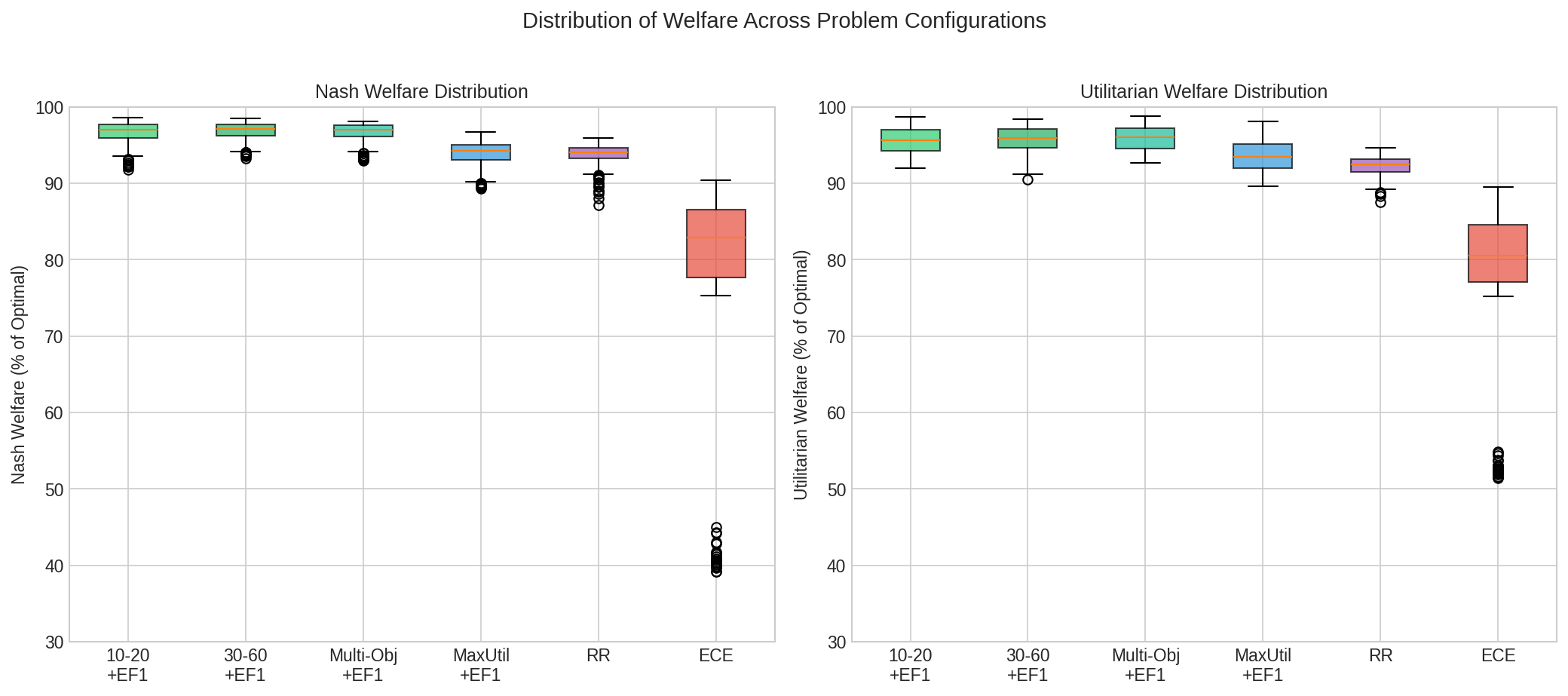}
  \caption{Distribution of welfare ratios across problem configurations. Boxplots visualize stability: learned models are tightly concentrated, while ECE exhibits heavy tails and large variance.}
  \label{fig:welfare_boxplots}
\end{figure*}

\begin{figure*}[t]
  \centering
  \includegraphics[width=\textwidth]{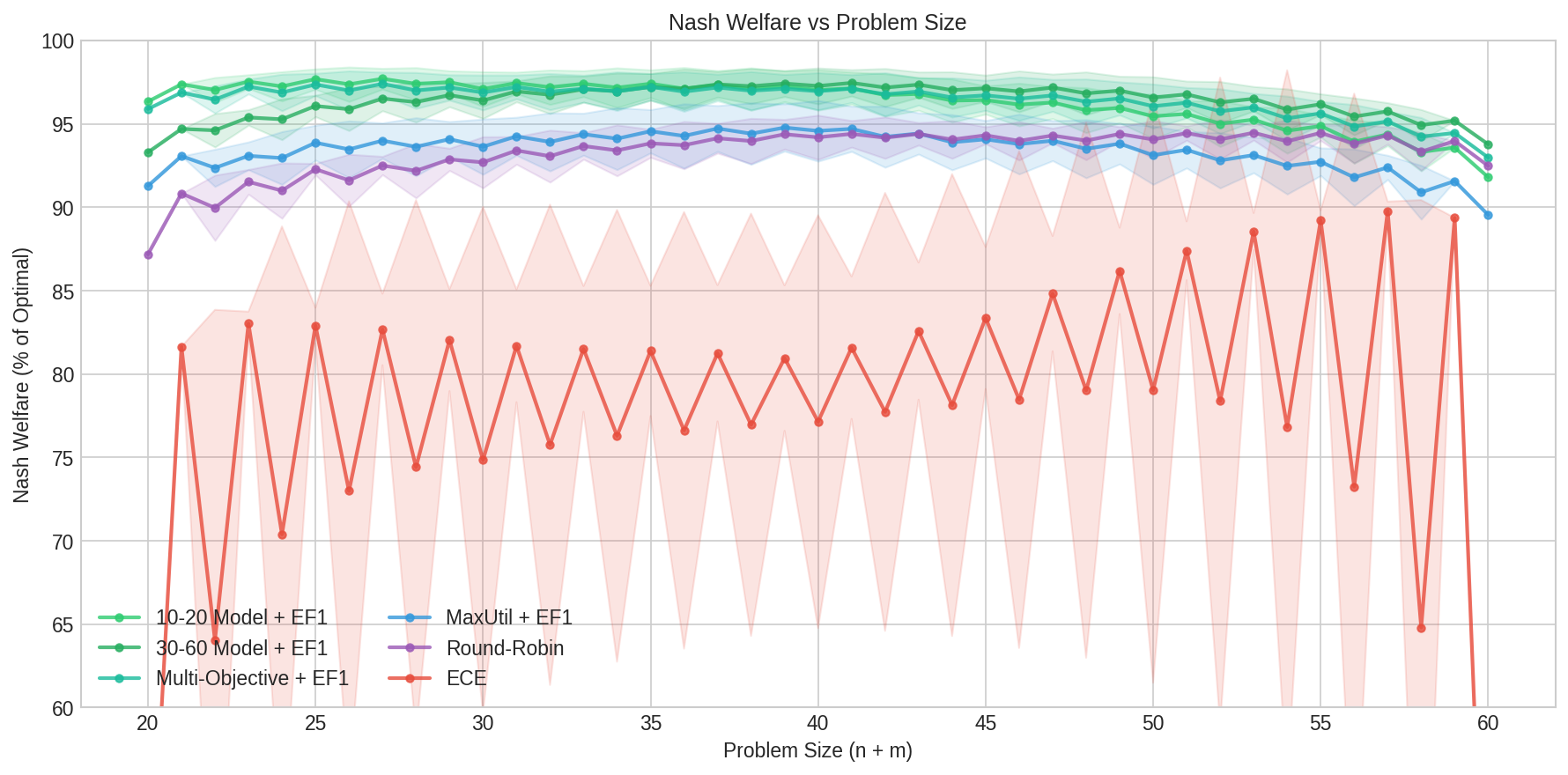}
  \caption{Nash welfare versus total problem size $(n+m)$. Learned models are robust across the evaluated scale range, while ECE shows pronounced instability.}
  \label{fig:nash_vs_size}
\end{figure*}

\begin{figure*}[t]
  \centering
  \includegraphics[width=\textwidth]{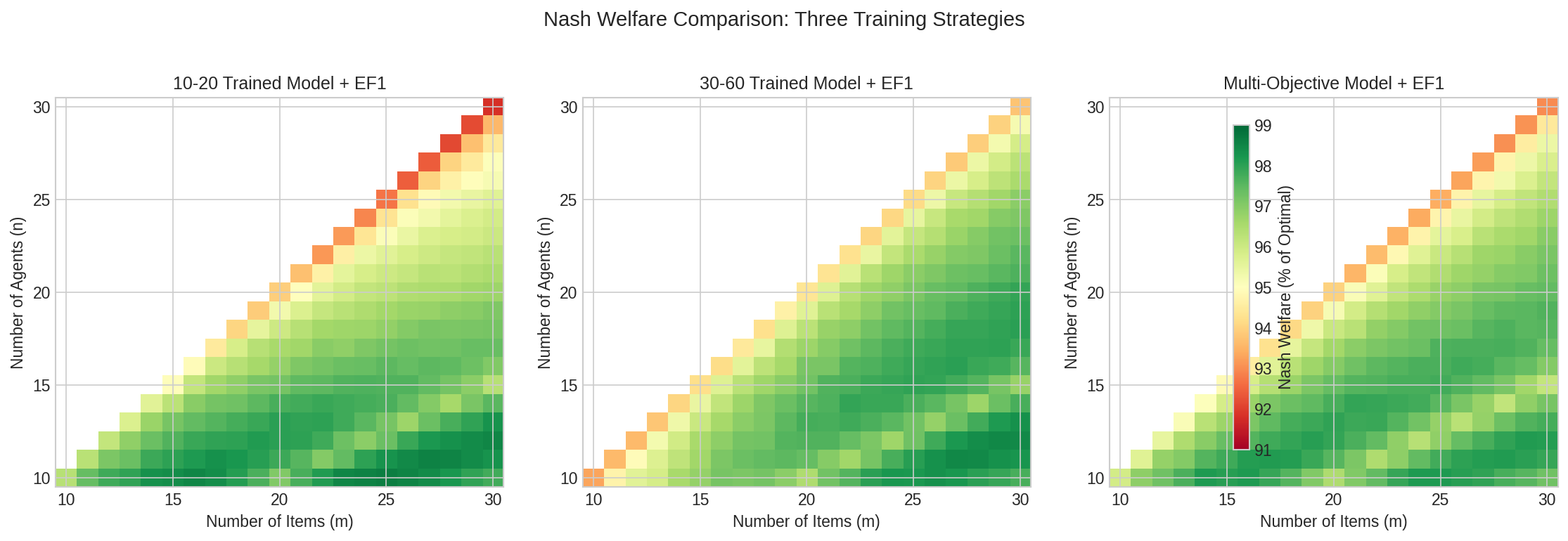}
  \caption{Nash welfare heatmaps for the three \FF training strategies across $(n,m)$, highlighting cross-size generalization. Color indicates Nash welfare as a percentage of the per-instance optimum.}
  \label{fig:heatmaps_models_only}
\end{figure*}

\begin{figure*}[p]
  \centering
  \includegraphics[width=0.95\textwidth]{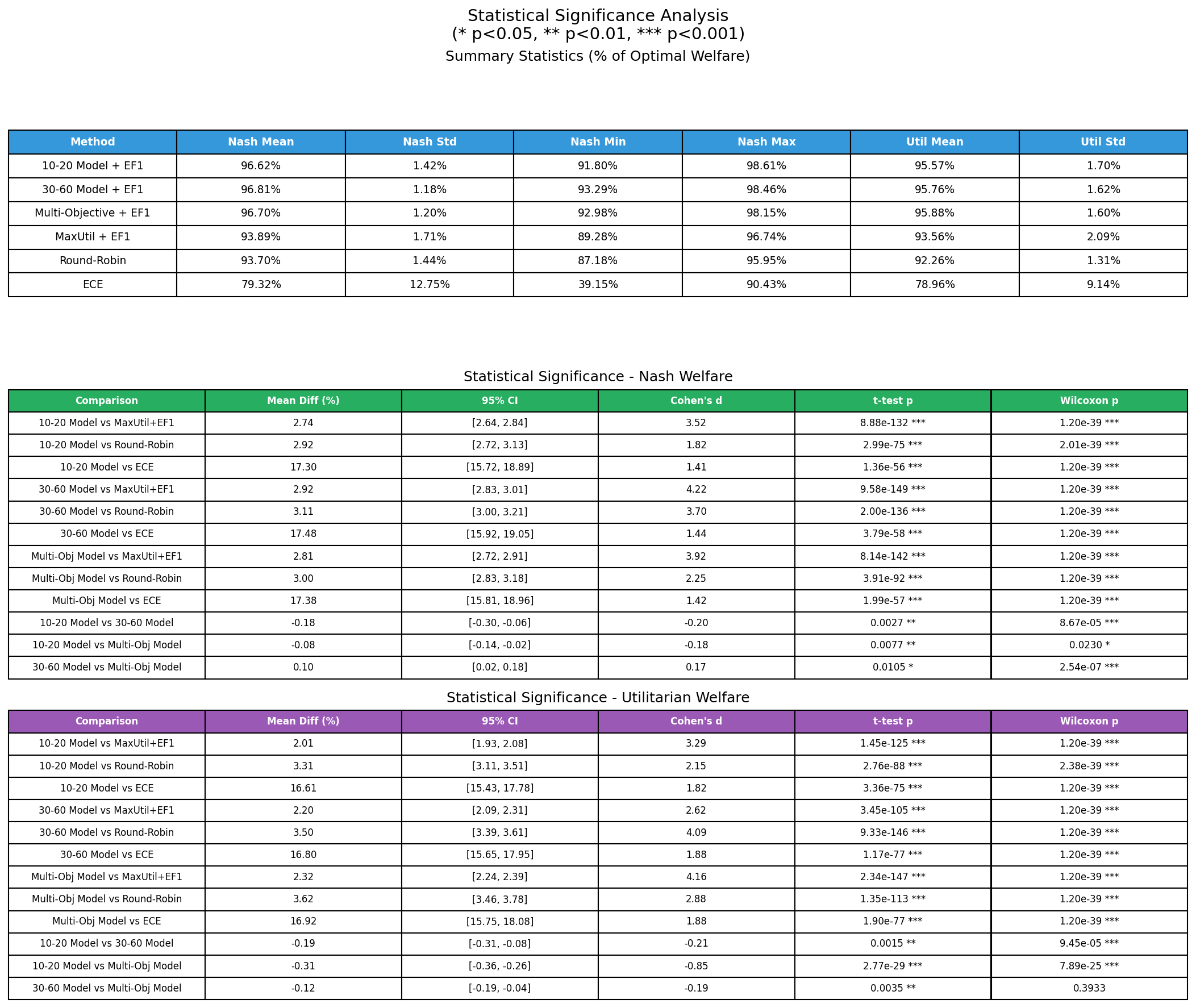}
  \caption{Summary statistics and pairwise significance tests for welfare ratios. We report mean differences, confidence intervals, effect sizes, and $p$-values for Nash and utilitarian welfare.}
  \label{fig:significance_tables}
\end{figure*}

\begin{table}[t]
\centering
\small
\begin{tabular}{lccccc}
\toprule
Algorithm & $10{\times}10$ & $10{\times}15$ & $10{\times}20$ & $20{\times}40$ & $20{\times}60$ \\
\midrule
\FF\ (10--20) & $0.81 \pm 0.39$ & $0.51 \pm 0.54$ & $0.78 \pm 0.50$ & $1.58 \pm 0.52$ & $1.42 \pm 0.49$ \\
\FF\ (30--60) & $0.99 \pm 0.07$ & $0.98 \pm 0.56$ & $1.08 \pm 0.52$ & $1.26 \pm 0.45$ & $1.19 \pm 0.47$ \\
\FF\ (Multi)  & $0.94 \pm 0.23$ & $0.77 \pm 0.52$ & $0.97 \pm 0.47$ & $1.40 \pm 0.49$ & $1.25 \pm 0.43$ \\
MaxUtil       & $1.00 \pm 0.00$ & $1.23 \pm 0.43$ & $1.36 \pm 0.52$ & $1.80 \pm 0.44$ & $1.56 \pm 0.52$ \\
\bottomrule
\end{tabular}
\caption{EF1-Repair passes until convergence (mean $\pm$ std). A value of 0 indicates the initial allocation already satisfied EF1 and no transfers were needed.}
\label{tab:ef1repair_passes_mean}
\end{table}

\begin{table}[t]
\centering
\small
\begin{tabular}{lccccc}
\toprule
Algorithm & $10{\times}10$ & $10{\times}15$ & $10{\times}20$ & $20{\times}40$ & $20{\times}60$ \\
\midrule
\FF\ (10--20) & 1/1/1/1 & 0/1/2/2 & 1/1/2/2 & 2/2/3/3 & 1/2/2/2 \\
\FF\ (30--60) & 1/1/1/1 & 1/2/2/3 & 1/2/2/2 & 1/2/2/3 & 1/2/2/3 \\
\FF\ (Multi)  & 1/1/1/1 & 1/1/2/2 & 1/1/2/2 & 1/2/2/2 & 1/2/2/2 \\
MaxUtil       & 1/1/1/1 & 1/2/2/2 & 1/2/2/3 & 2/2/3/3 & 2/2/3/3 \\
\bottomrule
\end{tabular}
\caption{EF1-Repair passes: percentiles (50th / 90th / 99th / max) across instances.}
\label{tab:ef1repair_passes_pct}
\end{table}

\FloatBarrier

\end{document}